\documentclass[%
onecolumn, 11pt,
superscriptaddress,
nofootinbib,
 amsmath,amssymb, amsthm,
pra,
showkeys
]{revtex4-2}

\usepackage[utf8]{inputenc}
\usepackage{amsfonts, amsmath, amssymb, amsthm}
\usepackage{dsfont, mathrsfs, bm}
\usepackage{enumerate}
\usepackage{graphicx}
\usepackage[table]{xcolor}
\usepackage{mdframed}
\usepackage{csquotes}

\usepackage{hyperref}
\hypersetup{
	colorlinks = true,
	linkcolor = blue,
	citecolor = blue,
	urlcolor = blue,
	filecolor = black,
	linktoc = all,
}

\newcommand{\tr}{\operatorname{tr}} 
\newcommand{\one}{\mathds{1}} 
\newcommand{\HH}{\mathscr{H}} 
\newcommand{\rr}{ {\hspace{1pt} \| \hspace{1pt}} } 
\newcommand{\RR}{\mathbb{R}}
\newcommand{\mcS}{\mathcal{S}}

\newcommand{\Lam}{\Lambda}

\newcommand{\ketbra}[1]{|#1\rangle \langle #1|}
\newcommand{\ket}[1]{|#1\rangle}

\newcommand{\MM}{M}
\newcommand{\M}{\MM}

\newcommand{\cM}{\mathcal{M}}
\newcommand{\cN}{\mathcal{N}}

\newcommand{\sM}{\mathscr{M}}

\newcommand{\eps}{\epsilon}
\newcommand{\lam}{\lambda}
\newcommand{\dimH}{d}

\newcommand{\diamondN}[1]{{\left\| {#1} \right\|_{\diamond}}}

\newcommand{\normN}[2]{{\left\| {#2} \right\|_{#1}}}
\newcommand{\pos}{\mathcal{P}}
\newcommand{\dist}{\gamma}

\newcommand{\geps}{g(\eps)}

\newtheorem{thm}{Theorem}
\newtheorem{lem}[thm]{Lemma}
\newtheorem{prop}[thm]{Proposition}
\newtheorem{cor}[thm]{Corollary}
\newtheorem{defn}[thm]{Definition}
\newtheorem{rem}[thm]{Remark}
\newtheorem{example}[thm]{Example}

\begin{document}

\title{Continuity bounds on observational entropy and measured relative entropies}

\author{Joseph Schindler}
\email{JosephC.Schindler@uab.cat}
\affiliation{%
F\'{\i}sica Te\`{o}rica: Informaci\'{o} i Fen\`{o}mens Qu\`{a}ntics, Departament de F\'{\i}sica, Universitat Aut\`{o}noma de Barcelona, 08193 Bellaterra, Spain}

\author{Andreas Winter}
\email{andreas.winter@uab.cat}
\affiliation{%
F\'{\i}sica Te\`{o}rica: Informaci\'{o} i Fen\`{o}mens Qu\`{a}ntics, Departament de F\'{\i}sica, Universitat Aut\`{o}noma de Barcelona, 08193 Bellaterra, Spain}
\affiliation{Instituci\'o Catalana de Recerca i Estudis Avan{\c{c}}ats, Pg.~Llu\'is Companys, 23, 08010 Barcelona, Spain}
\affiliation{Institute for Advanced Study, Technische Universit\"at M\"unchen, Lichtenbergstra{\ss}e 2a, 85748 Garching, Germany}

\date{5 August 2023}  



\begin{abstract}
We derive a measurement-independent asymptotic continuity bound on the observational entropy for general POVM measurements, making essential use of its property of bounded concavity. The same insight is used to obtain continuity bounds for other entropic quantities, including the measured relative entropy distance to a convex set of states under a general set of measurements. As a special case, we define and study conditional observational entropy, which is an observational entropy in one (measured) subsystem conditioned on the quantum state in another (unmeasured) subsystem. We also study continuity of relative entropy with respect to a jointly applied channel, finding that observational entropy is uniformly continuous as a function of the measurement. But we show by means of an example that this continuity under measurements cannot have the form of a concrete asymptotic bound.
\end{abstract}

\maketitle


\section{Introduction}

Observational entropy has recently emerged (in fact, re-emerged, \textit{cf.}~\cite{vonNeumann1929proof,vanKampen1954statistics,vonNeumann1955mathematical,wehrl1978general}) in studies of non-equilibrium statistical mechanics as a useful unifying framework to describe coarse-grained entropy in classical and quantum systems~\cite{safranek2019a,safranek2021brief,strasberg2020first}, with applications across thermodynamics and quantum information theory~\cite{safranek2019b,safranek2020classical,safranek2020quantifying,wehrl1979relation,riera2020finite,strasberg2021clausius,buscemi2022observational,amadei2019unitarity,zhou2022relations,zhou2022renyi,hamazaki2022speed,faiez2020typical,modak2022anderson,stokes2022nonconjugate,strasberg2022book,strasberg2022classicality,safranek2022work,sreeram2022chaos,zhou2023dynamical,safranek2023expectation}.

In the quantum case, for any measurement described by a POVM $\MM = (\MM_i)_{i \in \mathcal{I}}$, ($\MM_i \geq 0$, $\sum_i \MM_i = \one$), and quantum state described by a density matrix $\rho$, observational entropy
\begin{equation}
\label{eqn:oe}
    S_\MM(\rho) = -\sum_i p_i \log \frac{p_i}{V_i},
    \qquad
    p_i = \tr(M_i \rho),
    \qquad
    V_i = \tr(M_i),
\end{equation}
describes an entropy of the state $\rho$ ``coarse-grained'' by measurement $\MM$. The $p_i$ describe a probability distribution over ``macrostates'' (\textit{i.e.}~measurement outcomes), and the $V_i$ describe ``volumes'' of macrostates analogous to the phase space volumes defining Boltzmann entropy.

The formula \eqref{eqn:oe} defining observational entropy can be recast in relative entropy forms more suggestive of its information-theoretic content. In particular, with $\rho$ defined on a Hilbert space $\HH$ of finite dimension $d$, 
\begin{equation}
\label{eqn:mre-entropy}
    D(\Phi_M(\rho) \rr \Phi_M(\one/d)) = D(p \rr q) = \log d - S_\M(\rho),
\end{equation}
with $D$ denoting either quantum or classical relative entropy depending on its arguments. Here $M$ is implemented by quantum-classical channel
\begin{equation}
\label{eqn:mre-channel}
    \Phi_M(\sigma) = \textstyle\sum_i \tr(\sigma M_i) \ketbra{i},
\end{equation}
called the measuring channel of $M$, and \mbox{$p_i= \tr(M_i \rho)$}, $q_i=\tr(M_i \, \one /d)$ are the observed probability distributions induced by $M$ on $\rho$ and on the maximally mixed state, respectively.

It is evident from the definition that observational entropy is a continuous function of the state~$\rho$. For practical purposes, however, it is most useful to realize this continuity statement in terms of concrete continuity bounds, as has also been done for many other entropic quantities~\cite{fannes1973continuity,grabowski1977continuity,audenart2007sharp,audenart2005continuity,winter2016tight}. 

One naive version of such a bound follows quickly from other known bounds. Expanding the definition \eqref{eqn:oe} one can separate observational entropy into a sum of observed Shannon entropy $H(p)=-\textstyle\sum_i p_i \log p_i$ over measurement outcomes, and a mean Boltzmann entropy,
\begin{equation}
\label{eqn:shan+boltz}
    S_\MM(\rho) = H(p) + \textstyle\sum_i p_i \log V_i.
\end{equation}
Bounding the Shannon and Boltzmann terms separately one can obtain, given two states $\rho,\sigma$ such that $\frac{1}{2}||\rho - \sigma||_1 \leq \eps$, 
\begin{equation}
\label{eqn:naive-bound}
\begin{array}{rcl}
    \big|S_\M(\rho)-S_\M(\sigma)\big|  
    & \leq &
    \big|H(p^{\rho}) - H(p^{\sigma})\big| 
    + 
    \big|\textstyle\sum_i (p^{\rho}_i - p^{\sigma}_i) \log V_i \big|
    \\[10pt]
    & \leq &
    h(\eps) + \eps \, \log |\M| 
    + 
    \eps \, \big\| \log V_i \big\|_\infty .
\end{array}
\end{equation}
Here $|M|$ is the number of outcomes of the POVM $M$ (cardinality of the index set $\mathcal{I}$),  the function $h(x) = - x \log x - (1-x) \log (1-x)$ is the binary Shannon entropy, and $\| \log V_i \|_\infty = \max_i |\log V_i|$. The first two terms in the bottom line arise from the classical Zhang-Audenaert continuity bound on Shannon entropy~\cite{zhang2007estimating,wilde2011notes,winter2016tight}, while the latter term provides a loose bound on the Boltzmann term. Both terms are derived using elementary properties relating the quantum trace norm to probabilities~\cite{wilde2011notes}.

Though apparently useful insofar as it implies $|\Delta S|\to 0$ as $\eps \to 0$ for any particular $\M$, the bound \eqref{eqn:naive-bound} fails to be universal, depending on the measurement $\M$ both through the number of outcomes $|\M|$ and through the volume terms $V_i$. Moreover, it is unacceptably loose, as evidenced by the following pathology. Given any POVM $\M$, one can define another one $\M'$, with twice as many outcomes, by splitting each POVM element in half and counting it twice  ($\M_i \to \textrm{two copies } M_i/2$). Letting $\M \to \M'$ the observational entropies $S_\M(\rho)$ and $S_\M(\sigma)$ are both invariant, but by repeating this transformation the bound \eqref{eqn:naive-bound} can be made arbitrarily loose (noting that after $V_i \leq 1$ each iteration separately loosens both terms). This pathology highlights the need for an improved and universal bound.

One expects the shortcomings above to be common to any bound attempting to treat the Shannon and Boltzmann terms separately. On the other hand, one might hope that the presence of normalizing volume terms in the definition \eqref{eqn:oe} may actually help improve the bound compared to a Shannon entropy rather than (as occurs in \eqref{eqn:naive-bound}) making it looser. We will show this is indeed the case, improving the bound by treating the terms together.

In particular, the purpose of this note is to demonstrate that, for any measurement $\M$ and any states $\rho,\sigma$ such that $\frac{1}{2}||\rho - \sigma||_1 \leq \eps$, observational entropy obeys a continuity bound of the form
\begin{equation}
\label{eqn:oe-bound-intro}
    \big|S_\M(\rho)-S_\M(\sigma)\big| \leq \geps + \eps \, \log d.
\end{equation}
This bound is universal in that it is independent of $\M$, and depends only on the Hilbert space dimension $d$ and the trace distance between the states.

The function $g(\eps)$ arising in the bound is defined by $g(x) = - x \log x + (1+x) \log (1+x)$ for $x>0$, and $g(0) \equiv 0$. It derives here from the binary Shannon entropy $h(x) = - x \log x - (1-x) \log (1-x)$, where it arises as
\begin{equation}
\label{eqn:geps}
    g(\eps) = (1+\eps) \, h\Big(\frac{\eps}{1+\eps}\Big).
\end{equation}
Elsewhere it arises as an entropy in bosonic systems~\cite{holevo2001evaluating}. This $g(x)$ is particularly convenient, as throughout $x\in [0,\infty)$ it is continuous, monotonic increasing, non-negative, and concave.

The main bound \eqref{eqn:oe-bound-intro} appears in Theorem~\ref{thm:oe-continuity}. In order to establish it, we first will observe that observational entropy has the same bounded concavity property typical of other entropic quantities, that is, it is concave but not too concave. We then establish a general continuity bound based only on bounded concavity, obtaining the main result as a corollary. The method of proof follows closely the proof of Lemma~2 of \cite{winter2016tight}, which obtained a closely related tight bound on quantum conditional entropy, and which in turn is closely related to methods used in~\cite{alicki2004continuity,mosonyi2011renyi}, this method sometimes being known as the Alicki-Fannes-Winter (AFW) trick~\cite{alicki2004continuity,winter2016tight}. The general bound Proposition~\ref{thm:continuity-general} given here is closely related to Lemma~7 of \cite{winter2016tight}, and provides a slightly different generalization thereof. Continuity bounds of this form are often referred to as asymptotic continuity due to the bound per qudit scaling as $\eps \log d$ in the $n$-copy regime~\cite{donald2000continuity,synak2006asymptotic}.

In addition to the main observational entropy result, we also discuss continuity results for other entropic quantities that may be of general interest. In Proposition~\ref{thm:continuity-measured-single} we consider relative entropy distances of the form
\begin{equation}
    \inf_{\sigma \in \chi} D(\cM(\rho) \rr \cN(\sigma))
\end{equation}
to a convex set of states $\chi$ under channels $\cM,\cN$. In Theorem \ref{thm:continuity-measured} we then consider measured relative entropies under restricted measurements~\cite{piani2009relative}
\begin{equation}
    \inf_{\sigma \in \chi} \; \sup_{M \in \sM} D_{M}(\rho \rr \sigma),
\end{equation}
where $D_{M}(\rho \rr \sigma)$ is the classical relative entropy observed under measurement $M$. And in Proposition~\ref{thm:continuity-conditional} we apply these results to the case of an observational conditional entropy $S_{M_A}(A|B)_\rho$ motivated by relation to entropic uncertainty principles. All of these bounds are derived by appealing to a general bound in the form
\begin{equation}
    \big|Z(\rho)-Z(\sigma)\big| \leq \geps + \eps \, \kappa
\end{equation}
shown in Proposition~\ref{thm:continuity-general} for $Z(\rho)$ any function obeying bounded concavity/convexity. As a final step, in Corollary~\ref{thm:measurement-continuity-simulation-distance} we show that $S_M(\rho)$ and $D_M(\rho \rr \sigma)$ are continuous as functions of the measurement~$M$ (in an appropriate topology on POVMs). This result is a special case of Theorem~\ref{thm:measurement-continuity}, where we show uniform continuity of
\begin{equation}
    D\big(\Phi(\rho) \rr \Phi(\sigma)\big)
\end{equation}
as a function of channel $\Phi$. In Example~\ref{thm:example-binary} we show limitations preventing simple asymptotic bounds from being obtained for the measurement continuity.

In the following we continue to consider a finite dimensional Hilbert space $\HH$ of dimension~$d$ unless otherwise specified. Quantum states $\rho$ are positive semidefinite Hermitian operators normalized to unit trace. Measurements are described by POVMs (positive operator valued measures), defined as collections of positive semidefinite Hermitian operators summing to the identity, assumed here to have a finite number of outcomes. The von Neumann entropy is denoted by $S(\rho)=-\tr(\rho \log \rho)$, and Shannon entropy by $H(p) = -\sum_i p_i \log p_i$. The binary~Shannon entropy function is denoted by $h(x)$, and $g(x)$ is as above. Observational~entropy~$S_\M(\rho)$ is defined by \eqref{eqn:oe}.

\section{Bounded concavity}

An elementary property of Shannon entropy is bounded concavity~\cite{nielsen2010book} (in the case of classical Shannon entropy this property can be derived directly from a simple chain rule computation, but it can also be inferred as a subcase of the more commonly emphasized quantum generalization), which we state here as a lemma for later use.

\begin{lem}
\label{thm:shannon-concavity}
Let $\lam_k$ be a probability distribution and $p^{k}$ a probability distribution for each $k$. The Shannon entropy $H$ is concave but not too concave,
\begin{equation}
    0 \leq \Big[ H \big({\textstyle\sum}_k \lam_k \, p^k \big) - {\textstyle\sum}_k \lam_k \, H(p^k) \Big] \leq H(\lam).
\end{equation}
Note that $(p^k)_i = p_{i|k}$ is a conditional probability distribution and 
$\big(\sum_k \lam_k \, p^k\big)_i = \sum_k \lam_k \, p_{i|k}$.
\end{lem}

Bounded concavity of observational entropy follows from the bounded concavity of Shannon entropy.

\begin{lem}[OE Bounded Concavity]
\label{thm:oe-concavity}
Let $\lam_k$ be a probability distribution and $\rho_k$ a state for each~$k$. The observational entropy $S_\M$ is concave but not too concave,
\begin{equation}
    0 \leq \Big[ S_\M \big({\textstyle\sum}_k \lam_k \, \rho_k \big) - {\textstyle\sum}_k \lam_k \, S_\M(\rho_k) \Big] \leq H(\lam).
\end{equation}
\end{lem}

\begin{proof}
Let $\rho = \sum_k \lam_k \rho_k$ and $p_i = \tr(M_i \rho)$. Then $p_i = \sum_k \lam_k \, p_{i|k}$, where $p_{i|k} = \tr(M_i \rho_k)$. Observe that only the difference in Shannon entropy contributes to the difference in OE, since the mean Boltzmann entropy is equal for the two distributions. That is, we have (abusing the notation $H(p_i) \longleftrightarrow H(p)$ for simplicity where unambiguous)
\begin{align}
    \Delta S
    &\equiv
    S_\M(\rho) - {\textstyle\sum}_k \lam_k \, S_\M(\rho_k)
    \\
    &=
    H(p_i) + {\textstyle\sum}_i p_i \log V_i
    - {\textstyle\sum}_k \lam_k \, [ H(p_{i|k}) + {\textstyle\sum}_i p_{i|k} \log V_i]
    \\
    &=
    H(p_i) 
    - {\textstyle\sum}_k \lam_k \,  H(p_{i|k})
    \\
    &\equiv
    \Delta H.
\end{align}
But Lemma~\ref{thm:shannon-concavity} precisely says that
$0 \leq \Delta H \leq H(\lam)$,
the desired result.
\end{proof}

In general bounded concavity for real functions on convex sets may be defined as follows.

\begin{defn}[General Bounded Concavity/Convexity]
A function $Z: \mcS \to \RR$ defined on a convex set $\mcS$ is bounded concave if, for any finite convex combination $\rho = \sum_k \lam_k \rho_k$,
\begin{equation}
\label{eqn:Z-bounded-concave}
    0 
    \leq 
    \Big[ Z\big(\rho \big) - {\textstyle\sum}_k \lam_k \, Z(\rho_k) \Big]
    \leq 
    H(\lam).
\end{equation}
A function $Z$ is bounded convex if $-Z$ is bounded concave.
\end{defn}

In the next section we derive a continuity bound for functions of quantum states that applies to any function with a bounded concavity or convexity property of this form, with observational entropy continuity as a corollary.

\section{Continuity}

The continuity bound proved below is derived by appealing to bounded concavity. Similar methods have been used to derive continuity results for other entropic quantities in the literature, such as in~\cite{donald2000continuity,alicki2004continuity,synak2006asymptotic,mosonyi2011renyi,li2014relative,winter2016tight,shirikov2020advanced,bluhm2022continuity,bluhm2023general,shirikov2023quantifying}. The following decomposition is essential to performing the main trick of the proof, sometimes referred to as the Alicki-Fannes-Winter trick (\textit{cf.}~\cite{alicki2004continuity,winter2016tight}). Whereas the previous section only involved convex structure, the derivation of the following decomposition also involves topology via the trace norm. The resulting continuity bound depends on both the convex and topological structures of the space of quantum states.

\begin{lem}[$\omega \Delta$ Decomposition]
\label{thm:omega-delta-decomp}
    Let $\rho, \sigma$ be quantum states, and define $\eps = \frac{1}{2}||\rho-\sigma||_1$. There exists a state~$\omega$ and states $\Delta_\pm$ such that
    \begin{equation}
    \begin{array}{rcl}
        \omega
        & = &
        \dfrac{1}{1+\eps} \; \rho +   \dfrac{\eps}{1+\eps} \; \Delta_-
        \\[12pt]
        & = &
        \dfrac{1}{1+\eps} \; \sigma +   \dfrac{\eps}{1+\eps} \; \Delta_+.
    \end{array}
\end{equation}
\end{lem}

\begin{proof}
Let $\rho$, $\sigma$ be states such that $||\rho - \sigma||_1 = 2 \eps$, and let
\begin{equation}
    \eps \, \Delta_{\pm} = \pm (\rho - \sigma)_{\pm}. 
\end{equation}
The positive and negative parts $(\rho - \sigma)_{\pm}$  are defined in the usual way~\cite{wilde2011notes}, by diagonalizing the Hermitian matrix $\rho-\sigma$, then separating into the sum of a non-negative diagonal and a non-positive diagonal matrix. That is, $\rho-\sigma = U^\dag D U = U^\dag (D_+ + D_-) U = \eps \Delta_+ -  \eps \Delta_-$. By definition $\Delta_{\pm} \geq 0$ since they are Hermitian with only non-negative eigenvalues. Since $\rho-\sigma$ is traceless, we have $\tr(\Delta_+) = \tr(\Delta_-) = ||\Delta_\pm||_1$, with the last equality by positive semidefiniteness. Further, $||\rho-\sigma||_1 = ||D_+ + D_-||_1 = ||D_+||_1 + ||D_-||_1 = ||\eps \Delta_+||_1 + ||\eps \Delta_-||_1 = 2 \eps  ||\Delta_\pm||_1 = 2\eps $ using unitary invariance of the trace norm and the fact that $D_{\pm}$ are diagonal. 
Thus $\tr(\Delta_{\pm})=1$, and $\Delta_{\pm}$ are states. 

Now define
\begin{equation}
\label{eqn:decomp-omega}
    \omega 
    = \frac{1}{1+\eps} \; \frac{(\rho + \sigma)}{2} 
    + \frac{\eps}{1+\eps} \; \frac{(\Delta_+ + \Delta_-)}{2}.
\end{equation}
This $\omega$ is a convex combination of states, so $\tr(\omega)=1$ and $\omega$ is a state. The above relations taken together with $\omega$ can be viewed as expressing a linear decomposition
\begin{align}
    \label{eqn:decomp-difference}
    \rho - \sigma &= \eps \, (\Delta_+ - \Delta_-) 
    \\
    \rho + \sigma &= 2(1+\eps) \, \omega - \eps (\Delta_+ + \Delta_-).
\end{align}
From \eqref{eqn:decomp-omega} with \eqref{eqn:decomp-difference} it is straightforward to verify that
\begin{equation}
\label{eqn:continuity-omega-identities}
    \omega 
    = \frac{1}{1+\eps} \, \rho +   \frac{\eps}{1+\eps} \, \Delta_- 
    = \frac{1}{1+\eps} \, \sigma +   \frac{\eps}{1+\eps} \, \Delta_+,
\end{equation}
which completes the proof.
\end{proof}

Next appears a general continuity bound for any function $Z(\rho)$ obeying bounded concavity on quantum states. Since \eqref{eqn:bound-general-Z} is invariant under $Z \to -Z$ this implies the same bound for functions obeying bounded convexity.

\begin{prop}
\label{thm:continuity-general}
Let $Z(\rho)$ be any function satisfying bounded concavity (\textit{cf.}~\eqref{eqn:Z-bounded-concave}) on quantum states. We have the following continuity bound.

Let $\rho$, $\sigma$ be states such that $\frac{1}{2}||\rho - \sigma||_1 \leq \eps$. Then
\begin{equation}
\label{eqn:bound-general-Z}
    \big|Z(\rho)-Z(\sigma)\big| \leq \geps + \eps \, \kappa,
\end{equation}
where $\geps$ is as in \eqref{eqn:geps} and $\kappa = \sup_{(\mu,\nu)}|Z(\mu)-Z(\nu)|$ is the supremum absolute difference in $Z$ between any two states.
\end{prop}

\begin{proof}
Suppose $\eps = \frac{1}{2}||\rho - \sigma||_1$, and suppose there exist  a state $\omega$ and states $\Delta_\pm$ such that
\begin{equation}
\label{eqn:continuity-omega-identities-PROOF}
    \omega 
    = \frac{1}{1+\eps} \, \rho +   \frac{\eps}{1+\eps} \, \Delta_- 
    = \frac{1}{1+\eps} \, \sigma +   \frac{\eps}{1+\eps} \, \Delta_+.
\end{equation}
The existence of such states was demonstrated in Lemma~\ref{thm:omega-delta-decomp}. From \eqref{eqn:continuity-omega-identities-PROOF} we have two different convex decompositions $\omega = \lam \, \omega_1 + (1-\lam) \, \omega_2$, both with the same coefficients $\lam = 1/(1+\eps)$. By the assumption \eqref{eqn:Z-bounded-concave} of bounded concavity,
\begin{equation}
   \big[ \lam \, Z(\omega_1) + (1-\lam) \, Z(\omega_2) \big]
   \leq
    Z(\omega) 
   \leq 
    h(\lam) + \big[ \lam \, Z(\omega_1) + (1-\lam) \, Z(\omega_2) \big]
\end{equation}
holds for both of these convex decompositions. The trick is simply to use $\omega_1,\omega_2 = \rho,\Delta_-$ for one side of the inequality and $\omega_1,\omega_2 = \sigma,\Delta_+$ for the other. This amounts to
\begin{equation}
   \lam \, Z(\rho) + (1-\lam) \, Z(\Delta_-)
   \leq
    h(\lam)
    +
    \lam \, Z(\sigma) + (1-\lam) \, Z(\Delta_+).
\end{equation}
Rearranging yields
\begin{equation}
    \, (Z(\rho) - Z(\sigma)) 
   \leq
   \frac{h(\lam) +(1-\lam) \, (Z(\Delta_+)-Z(\Delta_-))}{\lam} .
\end{equation}
Finally we take the magnitude and apply triangle inequality,
\begin{align}
    |Z(\rho) - Z(\sigma)| 
   &\leq
   \frac{h(\lam)}{\lam} + \frac{(1-\lam)}{\lam} \, |Z(\Delta_+)-Z(\Delta_-)|
   \\
   & \leq
   \frac{h(\lam)}{\lam} + \frac{(1-\lam)}{\lam} \; \kappa.
\end{align}
The second inequality follows since a magnitude difference in $Z$ between two states is by definition less than or equal to $\kappa$. Finally with $\lam = 1/(1+\eps)$ and $h(\lam)=h(1-\lam)$ we obtain
\begin{equation}
    \, |Z(\rho) - Z(\sigma)| 
   \leq
   \geps + \eps \, \kappa .
\end{equation}
One can confirm the bound is monotonic in $\eps$. Therefore if it holds for $\eps= \frac{1}{2}||\rho - \sigma||_1$, as was just shown, then it also holds for $\eps' \geq \eps$. This completes the proof.
\end{proof}

In the following sections this general bound is evaluated for the observational entropy and extended to a form that includes several related quantities.

\section{Observational Entropy}

The general statement Proposition~\ref{thm:continuity-general} implies a continuity bound on observational entropy.

\begin{thm}[OE Continuity]
\label{thm:oe-continuity}
Let $\rho$, $\sigma$ be quantum states such that $\frac{1}{2}||\rho - \sigma||_1 \leq \eps$. Then observational entropy obeys the continuity bound
\begin{equation}
\label{eqn:oe-bound}
    \big|S_\M(\rho)-S_\M(\sigma)\big| \leq \geps + \eps \, \log d,
\end{equation}
where $g(\eps)$ is as in \eqref{eqn:geps} and $d$ is the Hilbert space dimension.
\end{thm}

\begin{proof}
By Lemma~\ref{thm:oe-concavity}, $S_\M$ obeys bounded concavity, so the bound follows from Proposition~\ref{thm:continuity-general}. An elementary property of observational entropy is that $0 \leq S_\M(\rho)\leq \log \dimH$~\cite{safranek2020quantifying}. Thus we have
$\kappa = \sup_{(\mu,\nu)}|S_\M(\mu)-S_\M(\nu)| 
= \sup_{\mu} S_\M(\mu) = \log \dimH$.
\end{proof}

The main result, a measurement-independent continuity bound on observational entropy, has therefore been established.

\vspace{2pt}

As an aside we note that bounds on both the Shannon and von Neumann entropies can be derived from the observational entropy continuity, as shown in the following corollary. The bounds thus obtained are looser than the known optimized forms~\cite{audenart2007sharp}.

\begin{cor}
Let $\rho$, $\sigma$ be states such that $\frac{1}{2}||\rho - \sigma||_1 \leq \eps$. Then von Neumann entropy $S$ obeys
\begin{equation}
     \big|S(\rho)-S(\sigma) \big| \leq \geps + \eps \, \log \dimH.
\end{equation}
Let $p_i$, $q_i$ be $N$-outcome probability distributions such that $\frac{1}{2}||\Vec{p} - \Vec{q}||_1 \leq \eps$. Then Shannon entropy~$H$ obeys
\begin{equation}
     \big|H(p_i)-H(q_i) \big| \leq \geps + \eps \, \log N.
\end{equation}
In this form these are both special cases of Theorem~\ref{thm:oe-continuity}, however, tighter bounds can be derived by other methods.
\end{cor}

\begin{proof}
We derive these as corollaries of Theorem~\ref{thm:oe-continuity}. 
For the quantum case suppose (wlg) that $S(\rho) \geq S(\sigma)$. Let $\M_\sigma$  measure in the $\sigma$ eigenbasis. Then 
$|S_{\M_\sigma}(\rho) - S_{\M_\sigma}(\sigma)| 
= S_{\M_\sigma}(\rho) - S(\sigma) 
\geq
S(\rho) - S(\sigma)
= |S(\rho) - S(\sigma)|.$
For the classical case, let $\rho= \sum_i p_i \ketbra{i}$ and $\sigma= \sum_i q_i \ketbra{i}$  in \mbox{$N$-dimensional} Hilbert space, and $\M_0$ measure in the $\ketbra{i}$ basis. (At least $N$ dimensions are needed to embed $p,q$ classically.) Since the states are diagonal we have $||\rho-\sigma||_1 = ||\Vec{p}-\Vec{q}||_1 \leq 2 \eps$. Then measuring $\M_0$ obtains $p_i$, $q_i$ and we have $|S_{\M_0}(\rho) - S_{\M_0}(\sigma)|
= |H(p_i)-H(q_i)|$.
\end{proof}

\section{Measured relative entropy}

The general Proposition~\ref{thm:continuity-general} likewise provides a continuity bound on measured relative entropies associated with a particular measurement, as well as on the quantum relative entropy.
\begin{prop}
\label{thm:continuity-measured-single}
    Let $\chi$ be a convex set of states (or more generally, of positive semidefinite Hermitian operators). Let $\cM$ and $\cN$ be arbitrary quantum channels. Then
    \begin{equation}
        Z(\rho) = \inf_{\sigma \in \chi} D(\cM(\rho) \rr \cN(\sigma)),
    \end{equation}is bounded convex. Thus $Z(\rho)$ obeys the bound \eqref{eqn:bound-general-Z} of Proposition~\ref{thm:continuity-general}. This ensures continuity whenever the maximum variation $\kappa$ is finite.
\end{prop}

\begin{proof}

Let $\rho = \sum_k \lam_k \rho_k$. First note that for any fixed $\sigma \in \chi$ we have
\begin{equation}
    \textstyle\sum_k \lam_k D(\cM(\rho_k) \rr \cN(\sigma)) - D(\cM(\rho) \rr \cN(\sigma)) = S(\cM(\rho)) - \textstyle\sum_k \lam_k S(\cM(\rho_k)),
\end{equation}
since $D(\cM(\rho) \rr \cN(\sigma)) = -S(\cM(\rho)) - \tr \cM(\rho) \log \cN(\sigma)$ and applying linearity of the channels. Thus
\begin{equation}
\label{eqn:single-sigma}
    0
    \leq
    \textstyle\sum_k \lam_k D(\cM(\rho_k) \rr \cN(\sigma)) - D(\cM(\rho) \rr \cN(\sigma))
    \leq
    H(\lam)
\end{equation}
by bounded concavity of von Neumann entropy. Note that $D(\cM(\rho) \rr \cN(\sigma)) =  \infty$ if and only if $\textstyle\sum_k \lam_k D(\cM(\rho_k) \rr \cN(\sigma))=\infty$ because the von Neumann terms are finite in finite dimension, and that \eqref{eqn:single-sigma} can be rearranged so the infinite case $\infty=\infty=\infty$ is well defined.

Now we introduce the infimum over $\sigma \in \chi$. It is immediate that the upper bound on convexity is retained, since
\begin{equation}
    \textstyle\sum_k \lam_k \inf_{\sigma_k \in \chi} D(\cM(\rho_k) \rr \cN(\sigma_k)) 
    \leq
    \inf_{\sigma \in \chi}  \textstyle\sum_k \lam_k D(\cM(\rho_k) \rr \cN(\sigma))
    \leq
    H(\lam) + Z(\rho),
\end{equation}
using that the sum of infima is no greater than the infimum of the sum and \eqref{eqn:single-sigma}.

The upper convexity bound just derived holds for the infimum of any family of bounded convex functions. In the particular case considered here, the lower convexity bound also holds, using convexity of the set $\chi$ and joint convexity of the relative entropy. In particular, since $Z(\rho)$ is an infimum, it follows that for any $\rho$ and any $\delta > 0$, there exists a state $\gamma \in \chi$ such that $D(\cM(\rho) \rr \cN(\gamma)) \leq Z(\rho) + \delta$. Thus fix $\delta > 0$ and let $D(\cM(\rho_k) \rr \cN(\gamma_k)) \leq Z(\rho_k) + \delta$, with $\nu = \sum_k \lam_k \gamma_k$. Then
\begin{equation}
        \textstyle\sum_k \lam_k \, Z(\rho_k) + \delta
        \geq
        \textstyle\sum_k \lam_k \, D(\cM(\rho_k) \rr \cN(\gamma_k))
        \geq
        D(\cM(\rho) \rr \cN(\nu))
        \geq
        Z(\rho).
\end{equation}
Since this holds for all $\delta > 0$, it follows that $\sum_k \lam_k \, Z(\rho_k) \geq Z(\rho)$. The first inequality expresses the initial infimum assumption, the second is joint convexity of relative entropy and linearity of the channels, and the last that $\nu \in \chi$ using convexity of the set.
\end{proof}

\begin{rem}
    If $\cN(\chi)$ contains an element of full rank the maximum variation $\kappa$ is necessarily finite. If $\cM = \cN$ and $\chi$ is a set of states containing the maximally mixed state, then $\kappa \leq \log d$.
\end{rem}

\begin{proof}
If $\cN(\chi)$ contains an element $\cN(\sigma)$ of full rank then all $|Z(\rho)| \leq \log d + ||\log \cN(\sigma)||_{\infty}$ and therefore $\kappa$ is finite. Meanwhile if $\cM = \cN$ and $\chi \ni \sigma$ is a set of states, then $D(\rho \rr \sigma) \geq D(\cM(\rho) \rr \cM(\sigma)) \geq 0$, so if also $(\one/d) \in \chi$ then $\kappa \leq \sup_\rho S(\rho) \leq \log d$.
\end{proof}

In particular, Proposition~\ref{thm:continuity-measured-single} in the case  $\cM = \cN = \Phi_M$  (\textit{cf.}~\eqref{eqn:mre-channel}) demonstrates a bound on measured relative entropy distance to a convex set of states, given the single measurement $M$. A natural question is whether this extends to include measured relative entropies associated with an allowed class of measurements $\sM$~\cite{piani2009relative,li2014relative},
\begin{equation}
    D_{\sM}(\rho \rr \sigma) = \sup_{M \in \sM} D(\Phi_{M} (\rho) \rr \Phi_{M}(\sigma)).
\end{equation}

In full generality we take $\chi$ a convex set of states and $\sM$ a set of measurements, and consider the quantity
\begin{equation}
    Z(\rho) = \inf_{\sigma \in \chi}  D_{\sM}(\rho \rr \sigma).
\end{equation}
Due to the supremum in $D_{\sM}$ this $Z(\rho)$ fails to obey bounded convexity, and the bound \eqref{eqn:bound-general-Z} cannot be applied directly. Nonetheless, we demonstrate below that a closely related continuity result can be obtained. This provides a more general and systematic extension to related results previously appearing in~\cite{li2014relative}, and also subsumes the entropic bounds already considered above.

A crucial step is the following minimax lemma.

\begin{lem}
\label{thm:minimax}
Let $\chi$ be a compact convex set of states. Let $\sM$ be a set of quantum channels that is closed under flagged finite convex combination, that is, such that if $\cM_k \in \sM$, and $\lam_k$ is a finite probability distribution, then also $\sum_k \lam_k \cM_k \otimes \ketbra{k} \in \sM$. Then
\begin{equation}
    \inf_{\sigma \in \chi} \sup_{\cM \in \sM} D(\cM(\rho) \rr \cM(\sigma) ) 
    =
    \sup_{\cM \in \sM}  \inf_{\sigma \in \chi} D(\cM(\rho) \rr \cM(\sigma) ).
\end{equation}
\end{lem}

\medskip
Operationally motivated classes $\sM$ of transformations often have the required property automatically, such as LOCC, SEP, PPT channels in composite systems. 

\begin{proof}
This lemma has been established previously in~\cite{brandao2013adversarial} on the grounds of a general minimax theorem of~\cite{farkas2006potential}. For completeness we include its demonstration here. We argue based on the minimax theorem in the form of Sion~\cite{sion1958minimax}. Namely, consider the function $\sum_i p_i D(\cM_i(\rho) \rr \cM_i(\sigma))$ of two arguments, one a state $\sigma \in \chi$, the other a finite-support probability distribution over channels from~$\sM$. The first set ($\chi$) is convex and compact by assumption, the second (denoted $\mathbb{P}\!\sM$) is convex by construction but otherwise an arbitrary space. For simplicity assume that our objective function $f(\sigma,p)$ is finite; it is then continuous, and it is convex in the first and linear---in particular concave---in the second argument. Hence, by~\cite{sion1958minimax} we get 
\[\begin{split}
\inf_{\sigma\in\chi} \sup_{\cM\in\sM} D(\cM(\rho) \rr \cM(\sigma)) 
  &= \inf_{\sigma\in\chi} \sup_{p\in\mathbb{P}\!\sM} \sum_i p_i D(\cM_i(\rho) \rr \cM_i(\sigma)) \\
  &= \sup_{p\in\mathbb{P}\!\sM} \inf_{\sigma\in\chi} \sum_i p_i D(\cM_i(\rho) \rr \cM_i(\sigma)).
\end{split}\]
Note that in general we need the sum over $i$ weighted by $p_i$’s, otherwise the final supremum is not large enough. But if the set of channels is closed under taking flagged convex combinations, i.e. with~$\cM_i$ it also contains $\cM' := \sum_i p_i \cM_i \otimes \ketbra{i}$, then $\sum_i p_i D(\cM_i(\rho) \rr \cM_i(\sigma)) = D(\cM'(\rho) \rr \cM'(\sigma))$, hence 
\[
  \sup_{p\in\mathbb{P}\!\sM} \inf_{\sigma\in\chi} \sum_i p_i D(\cM_i(\rho) \rr \cM_i(\sigma)) 
  = \sup_{\cM \in \sM}  \inf_{\sigma \in \chi} D(\cM(\rho) \rr \cM(\sigma) ),
\]
which concludes the proof in the case that $f$ is always finite. To cover the case of values $+\infty$, one can modify the argument, removing the infinities from the objective function; or use the minimax theorem from~\cite{farkas2006potential} as done in~\cite{brandao2013adversarial}.
\end{proof}

For classes $\sM$ of measurements closed under disjoint convex combination, and compact convex sets $\chi$ of states, we have---with the help of the minimax lemma---the following continuity theorem for measured relative entropies $Z(\rho)=\inf_{\sigma \in \chi}  D_{\sM}(\rho \rr \sigma)$. 

We remark that many commonly invoked measurement classes (such as LOCC, SEP, PPT measurements in composite systems) automatically fulfill the required closure property.  Furthermore, given some primitive set of measurements $\sM_0$ not already doing so, from an operational perspective it is perfectly feasible to complete this to an $\sM$ that is closed in this way, given the capability to choose which measurements to perform based on classical randomness.

\begin{thm}
\label{thm:continuity-measured}
Let $\chi$ be a compact convex set of states. Let $\sM$ be a convex set of measurements that is also closed under disjoint convex combination (that is, convex combinations taken on the disjoint union of the outcome sets). Let $\Phi_M$ as in \eqref{eqn:mre-channel} and define
\begin{equation}
    Z(\rho) = \inf_{\sigma \in \chi} \; \sup_{M \in \sM} \; Z_{M,\sigma}(\rho),
    \qquad \qquad
    Z_{M,\sigma}(\rho) = D(\Phi_M(\rho) \rr \Phi_M(\sigma) ).
\end{equation}
Then for any two states $\rho$, $\rho'$ such that $\frac{1}{2}||\rho - \rho'||_1 \leq \eps$, 
\begin{equation}
    \big|Z(\rho)-Z(\rho')\big| \leq \geps + \eps \, \kappa.
\end{equation}
Here
\begin{equation}
    \kappa = \sup_{M \in \sM} \left[ \sup_{\mu,\nu} \; \big|Z_{M}(\mu)-Z_{M}(\nu)\big| \right],
    \qquad \qquad
    Z_{M}(\rho) = \inf_{\sigma \in \chi} \;  Z_{M,\sigma}(\rho),
\end{equation}
and $g(\eps)$ is as in \eqref{eqn:geps}.
\end{thm}

\begin{proof}
Let $\Phi_{\sM}$ denote the set of measuring channels associated with class $\sM$, defined as in~\eqref{eqn:mre-channel}. Clearly 
$
Z(\rho) 
= \inf_{\sigma \in \chi} \; \sup_{M \in \sM} D(\Phi_M(\rho) \rr \Phi_M(\sigma) )
= \inf_{\sigma \in \chi} \; \sup_{\Phi \in \Phi_{\sM}} D(\Phi_M(\rho) \rr \Phi_M(\sigma) )
$.
Now the measuring channel is convex linear in $M$ on a fixed outcome set, $\Phi_{\sum_k \lam_k M_k} = \sum_k \lam_k \Phi_{M_k}$. To perform disjoint convex combination, which we denote $\oplus_k \lam_k M_k$, one first extends each outcome set to the disjoint union of outcome sets, padding with zeroes, then convexly combines on the fixed set. Suppose $\Phi_{M_k} \in \Phi_\sM$. It can be seen that $\Phi = \sum_k \lam_k \Phi_{M_k} \otimes \ketbra{k} = \Phi_{\oplus_k \lam_k M_k}$, which is also in $\Phi_\sM$ since $\oplus_k \lam_k M_k \in \sM$. Therefore the assumptions of the Lemma~\ref{thm:minimax} are fulfilled, and it holds
\begin{equation}
    \inf_{\sigma \in \chi} \; \sup_{M \in \sM} \; Z_{M,\sigma}(\rho)
    =
    \sup_{M \in \sM} \; \inf_{\sigma \in \chi} \; Z_{M,\sigma}(\rho).
\end{equation}
This provides the key step to relate $Z(\rho)$ to a form where bounded convexity can be exploited.

Proceeding, first using the minimax equality, and next that a difference of suprema is no less than the supremum difference, one obtains
\begin{align}
    | Z(\rho) - Z(\rho')|
    &=
    \left|\inf_{\sigma \in \chi} \; \sup_{M \in \sM} \; Z_{M,\sigma}(\rho) - \inf_{\sigma \in \chi} \; \sup_{M \in \sM} \; Z_{M,\sigma}(\rho') \right|
    \\
    & =
    \left| \sup_{M \in \sM} \; \inf_{\sigma \in \chi} \;  Z_{M,\sigma}(\rho) - \sup_{M \in \sM} \; \inf_{\sigma \in \chi} \; Z_{M,\sigma}(\rho') \right|
    \\ 
    & \leq
    \sup_{M \in \sM} \;  \left|Z_{M}(\rho)-Z_{M}(\rho') \right|.
\end{align}
But by Proposition~\ref{thm:continuity-measured-single} this $Z_M(\rho)$ is bounded convex and obeys \eqref{eqn:bound-general-Z}, hence further
\begin{align}
    | Z(\rho) - Z(\rho')|
    & \leq
    \sup_{M \in \sM} \;  \left[
     \geps + \eps \, \kappa_M
    \right]
    \\
    & =
     \geps + \eps \,  \sup_{M \in \sM} \; \kappa_M,
\end{align}
establishing the result.
\end{proof}

This provides asymptotic continuity of the entropies of restricted measurement first considered by Piani~\cite{piani2009relative}. Similar results were also obtained recently in independent work, in the context of local classes of measurements, and will appear elsewhere~\cite{lami2023notes}.

The main limitation of the Theorem~\ref{thm:continuity-measured} is the requirement that $\chi$ be convex. One may also be interested in measured relative entropy distances to non-convex sets of states such as coherent or classically correlated states. In this regard the Theorem can be compared to results of~\cite{li2014relative}, which obtained a looser bound but allows $\chi$ to be any set star-shaped around the maximally mixed state. Obtaining a bound for star-shaped $\chi$ with an optimal dimensional factor remains a useful future direction.

\section{Conditional observational entropy}

The above results are applied in this section to derive continuity bounds on an observational conditional entropy in bipartite systems, defined by
\begin{equation}
\label{eqn:conditional}
    S_{M_A}(A|B)_\rho  
    = -D\Big((\Phi_{M_A} \otimes I_B) \, \rho \, \Big\| (\Phi_{M_A} \otimes I_B)  (\one_A \otimes \rho_B) \Big),
\end{equation}
a quantity motivated by applications to entropic uncertainty principles with memory~\cite{coles2017entropic}. This conditional form describes quantum uncertainty in system $B$ having made a measurement $M_A$ in system $A$. 

We have the following bound. Note that the dimensional factor of $\log d$ arising from maximum variation is reduced by a factor of two from the corresponding bound on quantum conditional entropy $S(A|B)_\rho =-D(\rho || \one_A \otimes \rho_B)$~\cite{winter2016tight}.

\begin{prop}
\label{thm:continuity-conditional}
Let $\rho$, $\sigma$ be quantum states such that $\frac{1}{2}||\rho - \sigma||_1 \leq \eps$. Then
\begin{equation}
    \big|S_{M_A}(A|B)_\rho-S_{M_A}(A|B)_\sigma \big| \leq \geps + \eps \, \log d_A,
\end{equation}
where $g(\eps)$ is as in \eqref{eqn:geps} and $d_A$ is the Hilbert space dimension in system $A$.    
\end{prop}

\begin{proof}
We derive the bound as a consequence of Theorem~\ref{thm:continuity-measured-single}. To establish the link we note that the observational conditional entropy $S_{M_A}(A|B)_\rho$ defined by \eqref{eqn:conditional} admits the variational formula
\begin{equation}
\label{eqn:conditional-variational}
    S_{M_A}(A|B)_\rho =  \log d_A
    - \inf_{\omega_B} \; D\Big((\Phi_{M_A} \otimes I_B) \, \rho \, \Big\| (\Phi_{M_A} \otimes I_B)  (\tfrac{\one_A}{d_A} \otimes \omega_B) \Big),
\end{equation}
where the infimum is over all states $\omega_B$ in system $B$. 
(We will return to the justification of this variational formula after concluding the rest of the proof.) In light of this equality, the difference $S_{M_A}(A|B)_\rho-S_{M_A}(A|B)_\sigma$ is precisely equal to a difference between infimized relative entropies of the form in Theorem~\ref{thm:continuity-measured-single}, in particular with $\cM=\cN = \phi_{M_A} \otimes I_B$ and $\chi = \{\tfrac{\one_A}{d_A} \otimes \omega_B \, | \, \omega_B {\rm \ a \ state \ on \ }B \}$. Thus the Theorem~\ref{thm:continuity-measured-single} demonstrates the bound \eqref{eqn:bound-general-Z} holds, and it only remains to evaluate the maximum variation $\kappa$.

To determine $\kappa$ it suffices to obtain the bound $0 \leq S_{M_A}(A|B)_\rho \leq \log d_A$ for all $\rho$. From non-negativity of the relative entropy, the upper bound $S_{M_A}(A|B)_\rho \leq \log d_A$ follows immediately. To establish the lower bound we will make use of the following observations:
\begin{enumerate}[(i)]
    \item With $\Phi_M$ a measurement channel as in \eqref{eqn:mre-channel}, $(\Phi_M \otimes I) \rho = \sum_j \ketbra{j} \otimes p_j \, \rho_j$ yields a classical-quantum state, with $p_j=\tr(M_j \rho_A)$ the induced probability distribution over $M$ outcomes, and $\rho_j = \tr_A((M_j \otimes \one) \rho)/p_j$ some conditional states.
    \item For any two such classical-quantum states, $D(\rho||\rho') = D(p_j || p'_j) + \sum_j p_j D(\rho_j||\rho'_j)$.
    \item With $p_j=\tr(M_j \rho_A)$ and $q_j = \tr(M_j)/d_A$, one has $S_M(\rho_A) = \log d_A - D(p_j || q_j)$ the marginal observational entropy in system $A$.
    \item With $p_j$, $\rho_j$ as in (i), one has $\sum_j p_j \rho_j = \rho_B$ since $\sum_j M_j = \one_A$.
    \item $S_{M_A}(A|B)_\rho$ is concave in $\rho$.
\end{enumerate}
We thus have
\begin{align}
    S_{M_A}(A|B)_\rho 
    &=
    \log d_A
    - D\Big((\Phi_{M_A} \otimes I_B) \, \rho \, \Big\| (\Phi_{M_A} \otimes I_B)  (\tfrac{\one_A}{d_A} \otimes \rho_B) \Big)
    \\[4pt]
    &=
    \log d_A
    - D(p_j || q_j) - \sum_j p_j D(\rho_j||\rho_B)
    \\
    \label{eqn:conditional-magic-formula}
    &=
    S_{M_A}(\rho_A) - S(\rho_B) + \sum_j p_j S(\rho_j).
\end{align}
The first equality is by definition, the second uses (i) and (ii), and the third uses (iii) and (iv). Finally, we note that for pure states $\ket{\psi_{AB}}$ we have that $S(\rho_A)=S(\rho_B)$ and therefore that $S_{M_A}(\rho_A) \geq S(\rho_B)$. Thus for pure states we have $S_{M_A}(A|B)_{\ketbra{\psi}} \geq 0$. But by concavity this non-negativity extends to convex combinations and so in general it holds that $S_{M_A}(A|B)_{\rho} \geq 0$.

To complete the proof, note that since $0 \leq S_{M_A}(A|B)_\rho \leq \log d_A$, the same bounds hold for the infimized relative entropy in \eqref{eqn:conditional-variational}. Thus the maximum variation $\kappa$ in Theorem~\ref{thm:continuity-measured-single} is given by $\kappa = \sup_{\mu,\nu} |Z(\mu)-Z(\nu)| \leq \sup_{\mu} Z(\mu) \leq \log d_A$.

It remains to justify the variational formula \eqref{eqn:conditional-variational}. Applying the same methods used to obtain~\eqref{eqn:conditional-magic-formula} earlier, \eqref{eqn:conditional-variational} can be seen to become $S_{M_A}(\rho_A) + \sum_j p_j S(\rho_j) - \inf_{\omega_B} \; \tr(-\rho_B \log \omega_B)$. But it follows immediately from $D(\rho_B \rr \omega_B) \geq 0$ that $\tr(-\rho_B \log \omega_B) \geq \tr(-\rho_B \log \rho_B)$ for any $\omega_B$, establishing that $\omega_B = \rho_B$ saturates the infimum.
\end{proof}

\section{Continuity with respect to measurements}

The preceding sections have demonstrated concrete uniform continuity bounds on observational entropy $S_M(\rho)$ (and related quantities) under variations in the state $\rho$. An equally important question is that of continuity under changes in the measurement $M$. With a suitable norm in the space of measurements (namely the diamond norm between measuring channels), can similar bounds be obtained on $|S_{M}(\rho)-S_{M'}(\rho)|$ for a fixed state?

We will show that although the observational entropy $S_M(\rho)$ is indeed a continuous function of $M$, no simple concrete bounds of the type obtained earlier can hold. A strong limitation is exhibited by the following example.

\begin{example}
\label{thm:example-binary}
\normalfont
Let $M = (\Pi, \one - \Pi)$ be a binary measurement, and $\Tilde{M} = (\one-\Pi, \Pi)$ its permutation. Consider a convex combination of this binary $M$ with its permutation,
\begin{equation}
    M_\lam = (1-\lam)M + \lam \Tilde{M}.
\end{equation}
This has POVM elements $M_{\lam} = \Big(\big[(1-2\lam)\Pi + 2\lam (\one/2)\big], \big[(1-2\lam)(\one-\Pi) + 2\lam (\one/2)\big] \Big)$. It follows that also
\begin{equation}
    M_\lam = (1-2\lam) \, M + 2\lam \, \left[\tfrac{\one}{2}\right],
\end{equation}
where $\left[\tfrac{\one}{2}\right]$ denotes the trivial binary measurement with elements $(\one/2,\one/2)$. This shows that mixing $M$ with $\lam$ of its permutation is equivalent to mixing $M$ with $2 \lam$ of pure noise.

Observe that measuring channels are convex linear in the measurement, namely 
$\Phi_{(1-\lam)M+\lam N} = (1-\lam)\Phi_M+\lam \Phi_N$. Thus the diamond norm distance between channels implementing $M_\lam$ and $M$ is bounded by 
$
\frac{1}{2}\| \Phi_{M_\lam} - \Phi_M \|_\diamond 
= 
\frac{1}{2}\| \lam (\Phi_{\tilde{M}} - \Phi_M) \|_\diamond 
\leq
\lam
$
due to the universal bounds on trace distances between states~\cite{wilde2011notes}. This provides a strong sense in which $M_\lam$ and $M$ are close together for continuity purposes, and ensures in particular that the probability distributions  (outcome statistics) induced by $M_\lam, M$ on any state differ by no more than $\lam$ in the total variational distance (half of the $\ell_1$ norm).

Continuing, suppose $\rho=\ketbra{\psi}$ is pure and $M = (\ketbra{\psi}, \one - \ketbra{\psi})$ measures $\rho$ perfectly. The observational entropy of $M$ and its permutation are equal in general, and in the present case one finds $S_M(\rho) = S_{\Tilde{M}}(\rho) = 0$. On the other hand, for the convex combination $M_\lambda$ one has probabilities and volumes
\begin{equation}
    \begin{array}{lcl}
         p_0 = 1-\lam, & \hspace*{10mm} & V_0 = 1+\lam(d-2), \\[4pt]
         p_1 = \lam,  &  & V_1 = (d-1) - \lam(d-2),
    \end{array}
\end{equation}
leading to an observational entropy
\begin{equation}
    S_{M_\lam}(\rho) 
    = (1-\lam) \log \left( \frac{1+\lam(d-2)}{1-\lam} \right)
    + \lam \; \log \left( \frac{(d-1) - \lam(d-2)}{\lam} \right).
\end{equation}
This can be viewed as Boltzmann and Shannon terms by expanding the logs.

Now the entropy $S_{M_\lam}(\rho)$ may be used to analyze both the bounded concavity, and directly the continuity, of observational entropy viewed as a function of measurement. First consider the bounded concavity. In the present case we have an entropy concavity difference 
\begin{equation}
\label{eqn:example-concavity-bound}
    \Big| S_{(1-\lam)M + \lam \Tilde{M}}(\rho) - \big[(1-\lam) S_{M}(\rho) + \lam S_{\Tilde{M}}(\rho)\big] \Big| = S_{M_\lam}(\rho)
\end{equation}
since $M_{\lam} = (1-\lam)M + \lam \Tilde{M}$ with the $M,\Tilde{M}$ each giving zero entropy.
And second consider directly the continuity. In the present case we have measurements such that $\frac{1}{2}\|\Phi_{M_{\lam}} - \Phi_M\|_\diamond \leq \lam$ and whose entropy difference is
\begin{equation}
\label{eqn:example-continuity-bound}
    \big| S_{M_{\lam}}(\rho) - S_{M}(\rho) \big| = S_{M_\lam}(\rho)
\end{equation}
since the latter term gives zero. Therefore general bounds on continuity and concavity can be no better than in this particular case.

\pagebreak\noindent
We make the following observations (with $S_\lam \equiv S_{M_\lam}(\rho)$ and $d \geq 2$ and $\lambda \in [0,1/2]$):
\begin{itemize}
    \item Any $\lam > 0$ implies $S_\lam>0$. For fixed $d$, as $\lam \to 0$ also $S_\lam \to 0$, consistent with continuity.
    \item $S_{1/2} = \log d$, so \eqref{eqn:example-concavity-bound} shows that the naive concavity bound $H(\lam)$ that held for states (\textit{cf.}~\eqref{eqn:Z-bounded-concave}) can be violated by a maximum amount.
    \item For fixed $\lam > 0$, one can achieve $S_\lam / \log d$ arbitrarily close to $1$ by taking $d$ sufficiently large. This shows that there is no continuous $f(\lam)$ such that $f(0)=0$ and $S_\lam \leq f(\lam) \log d$.
\end{itemize}
\end{example}

\noindent
The final observation implies a strong limitation on asymptotic continuity bounds for measurements: there cannot be any continuous function $f(\lam)$ with $f(0)=0$ such that
\begin{equation}
\label{eqn:example-no-go}
    \big| S_{N}(\rho) - S_{M}(\rho) \big| \leq f(\lam) \log d
\end{equation}
for all $\frac{1}{2}\|\Phi_N - \Phi_M\|_\diamond \leq \lam$. This form includes all bounds of the type considered earlier for continuity of various entropies with respect to states.

Moreover, one might have hoped to approach the measurement continuity problem similarly to the method used for states, by attempting some form of an AFW trick (modified for channels) based on bounded concavity. However, the example shows that not only does the naive concavity bound by $H(\lam)$ fail, but also (likewise with the continuity above) that no concavity bound of the form $f(\lam) \log d$ is possible.

Due to these limitations, at present we settle for the following more abstract statement of continuity over measurements. We give here a general statement containing observational entropy as a special case.

\begin{thm}
\label{thm:measurement-continuity}
    Let $\rho,\sigma$ be states such that $D(\rho \rr \sigma) < \infty$. Then $F(\Phi) =  D\big(\Phi(\rho) \rr \Phi(\sigma)\big)$ is a uniformly continuous function of the quantum channel $\Phi$.
\end{thm}

\begin{proof}
    We are in finite dimension, so the quantum channels form a finite-dimensional convex set and all norms induce the same topology. For concreteness we use the diamond norm. 
    To show continuity we have to show that for any $\Phi$ and any $\delta > 0$, there exists an $\eps > 0$ such that 
    $\diamondN{\Phi - \Phi'} < \eps$ 
    implies 
    $|F(\Phi) - F(\Phi')|<\delta$.

    To begin, for any channel $\Phi$ define $\Phi_s = (1-s) \Phi + s \tau$, which mixes $\Phi$ with noise, letting $\tau(X) = \tr(X) \one /d$ prepare maximally mixed state. Observe that $F(\Phi)$ is convex in $\Phi$ and is finite and non-negative, and that $F(\tau)=0$. It follows that $F(\Phi) \geq (1-s) F(\Phi) \geq F(\Phi_s)$. Therefore for any $\Phi$ one obtains $|F(\Phi) - F(\Phi_s)| \leq s F(\Phi)$.

    Additionally, for any $\Phi$, the states $\Phi_s(\rho)$ and $\Phi_s(\sigma)$ are positive definite. The set $\pos$ of positive definite matrices is an open set in the $2$-norm matrix topology (the usual Euclidean topology on matrix elements). Over the domain $\mu,\nu \in \pos$, the function $D(\mu \rr \nu)$ is continuous, by the elementary reasoning that it is a composition of matrix multiplication, linear combination, trace, and log, which are all continuous functions on the positive set $\pos$. Thus for any $\delta_s > 0$ there exists an $\eps_s >0$ such that if both $\normN{2}{\Phi_s(\rho) - \Phi'_s(\rho)}<\eps_s$ and $\normN{2}{\Phi_s(\sigma) - \Phi'_s(\sigma)}<\eps_s$, then $|F(\Phi_s) - F(\Phi'_s)|<\delta_s$. This~$\eps_s$ must be chose sufficiently small so that both $\eps_s$-balls lie entirely in $\pos$.  

    Meanwhile, for channels $\Phi,\Phi'$ it holds 
    $\diamondN{\Phi_s - \Phi'_s} = \diamondN{(1-s) (\Phi - \Phi')} \leq \diamondN{\Phi - \Phi'} $. And therefore for any state $\nu$ we have
    $\normN{2}{\Phi_s(\nu) - \Phi'_s(\nu)}
    \leq
    \normN{1}{\Phi_s(\nu) - \Phi'_s(\nu)}
    \leq
    \diamondN{\Phi_s - \Phi'_s}
    \leq
    \diamondN{\Phi - \Phi'}
    $.

    Now we combine the ingredients. For any $1>s>0$, $\delta_s > 0$ let $\epsilon_s$ be as above, and suppose $\diamondN{\Phi - \Phi'} < \eps_s$. Observe that
    \begin{align}
        |F(\Phi) - F(\Phi')| 
        & = 
        |F(\Phi) - F(\Phi_s) + F(\Phi_s) - F(\Phi'_s) + F(\Phi'_s) - F(\Phi')|
        \\
        & \leq
        |F(\Phi) - F(\Phi_s)| + |F(\Phi_s) - F(\Phi'_s)| + |F(\Phi'_s) - F(\Phi')|
        \\
        & \leq
        s F(\Phi) + \delta_s + s F(\Phi')
        \\
        & \leq
        \delta_s + 2 s D(\rho \rr \sigma).
    \end{align}
    The first inequality is triangle inequality. The second follows from the observations listed above. The last is relative entropy monotonicity.

    Finally, choose a $\delta > 0$. Since we assumed $D(\rho \rr \sigma) < \infty$, we can choose $s>0$ sufficiently small so that $2 s D(\rho \rr \sigma) < \delta$. Next choose any $\delta_s > 0$ sufficiently small so that also  $2 s D(\rho \rr \sigma) + \delta_s < \delta$. There exists an $\epsilon_s >0$ as above so that $\diamondN{\Phi - \Phi'} < \eps_s$ implies $|F(\Phi_s) - F(\Phi'_s)| < \delta_s$ and therefore also that $|F(\Phi) - F(\Phi')| < \delta$. This completes the proof of continuity. The finite dimensional quantum channels form a compact set, thus continuity of $F$ implies it is also uniformly continuous. One can also argue directly (but with more technical difficulty) for uniform continuity by fixing $s=\delta/4D(\rho \rr \sigma)$ and using uniform continuity of $\log$ on the domain $[s/d,\infty)$.
\end{proof}

Observational entropy is determined by $D(\Phi(\rho) \rr \Phi(\sigma))$ in the case $\sigma = \one/d$ and $\Phi$ a measuring channel. Since $D(\rho \rr \one/d) < \infty$ always, this implies that  (in the topology induced by the diamond norm on $\Phi_M$) observational entropy $S_M(\rho)$ is a continuous function of the measurement. In this form the continuity is useful only for measurements on the same outcome set.

The continuity under measurements can be extended to measurements on arbitrary outcome sets by defining an appropriate topology. Such a topology can be induced by a \textit{simulation distance} defined (for POVMs $M,N$) by
\begin{equation}
    \dist(M,N) = \tfrac{\Vec{\dist}(M,N) + \Vec{\dist}(N,M)}{2},
    \qquad \qquad
    \Vec{\dist}(M,N) = \inf_{\Lambda} \; \tfrac{1}{2} \| \Phi_{\Lambda M} - \Phi_N \|_{\diamond},
\end{equation}
where the infimum defining the ``one-way'' simulation distance $\Vec{\dist}(M,N)$ is over stochastic maps $\Lambda$ from the $M$ outcome set to the $N$ outcome set, so that $\Lambda M$ and $N$ are measurements with the same set of outcomes. The function $\dist(M,N)$ can be shown to obey the axioms of a (pseudo-)metric, and therefore induces a topology on the space of POVMs. This distance captures how well two POVMs can approximate one another using classical postprocessing by stochastic channels---if both can emulate each other perfectly by postprocessing alone, they are effectively the same POVM. In terms of this topology we have continuity of measured relative entropy $D_M$ with respect to $M$, and therefore also of $S_M$ the observational entropy.

\begin{cor}
\label{thm:measurement-continuity-simulation-distance}
    Let $\rho,\sigma$ be states such that $D(\rho \rr \sigma) < \infty$. Then $F(M) =  D_M(\rho \rr \sigma)$ is a continuous function of the measurement $M$ in the simulation distance topology. This also implies continuity of $F_0(M) = S_M(\rho)$ as a particular case. The continuity is uniform (at least for $d < \infty$).
\end{cor}

\begin{proof}
    Fix $\delta > 0$ and let $\eps = \eps(\delta)$ as supplied by Theorem~\ref{thm:measurement-continuity}. Choose any $c$ with $\eps > c >0$, and let $\eps' = (\eps -c)/4 >0$. Suppose $\gamma(M,N)<\eps'$. From the definition of $\gamma$, there exist $\Lambda, \Lambda'$ such that $\| \Phi_{\Lambda M} - \Phi_N \|_{\diamond} < \eps$ and $\| \Phi_{\Lambda' N} - \Phi_M \|_{\diamond} < \eps$. Thus $|D_{\Lambda M} - D_N|<\delta$ and $|D_{\Lambda' N} - D_M|<\delta$. We then have $D_M \geq D_{\Lam M} > D_N - \delta$ and also $D_M < D_{\Lam' N} + \delta \leq D_N + \delta$, which follow using relative entropy monotonicity, and thus $|D_M - D_N| < \delta$. This completes the proof.
    
    Put more simply, we stated that for small $\gamma(M,N)$ we have
    \mbox{$D_N \geq D_{\Lam' N} \approx D_M \geq D_{\Lam M} \approx D_N$},
    which implies $D_M \approx D_N$, 
    where $\geq$ are by monotonicity and $\approx$ are by the earlier continuity.
\end{proof}

\section{Concluding remarks}

The main Theorem~\ref{thm:oe-continuity} described a continuity bound \eqref{eqn:oe-bound} on observational entropy that is universal, in the sense of being independent of the POVM $\M$ defining the coarse-graining, and with a dimensional factor $\log d$ of the Hilbert space dimension.

This bound improves upon the naive bound \eqref{eqn:naive-bound} insofar as it scales only with $\log d$ independently of the number of measurement outcomes $|M|$. This reflects the intuitive status of observational entropy as a classical entropy in quantum systems, compared, for example to the Shannon entropy of an observable. The improved universality does come with a small trade-off, as when $M$ is a projective measurement in a complete orthonormal basis, the naive bound \eqref{eqn:naive-bound} is slightly stronger. We note that~\eqref{eqn:oe-bound}~can only be relevant when $\rho,\sigma$ have overlapping support. In case the supports are orthogonal, the trace distance $\eps=1$ is maximized, and \eqref{eqn:oe-bound} is weaker than the trivial $|\Delta S_M| \leq \log d$. This demonstrates that \eqref{eqn:oe-bound} cannot be tight in general---for fixed $\eps$, there is not generally a choice of $\rho,\sigma,\M$ for which the inequality is saturated. 

Meanwhile, in the case of continuity with respect to measurements, it was seen that although observational entropy is continuous, explicit bounds on the convergence are difficult (if not impossible) to obtain. Other notable limitations on the present results were the necessity of a convex reference set for measured relative entropy continuity (as discussed earlier), and  the restriction throughout the paper to finite dimensional spaces. Extending our results to the infinite dimensional case encounters difficulties associated with the possibility of infinite Boltzmann and von Neumann entropies, and doing so carefully will be a useful continuation.

One can also consider a classical version of the observational entropy development for density distributions $\rho,\sigma$ on phase space $\Gamma$. The trace is an integral on $\Gamma$, the norm the $L^1$ norm, and the measurement $M$ a partition of unity. The dimensional factor $d$ (better labelled $V$ in the classical case) is the total volume of phase space in an appropriately chosen measure. The resulting bound~\eqref{eqn:oe-bound} holds equally well in the classical case so long as one can restrict consideration to a finite region of phase space.

As with the quantum case, the classical \eqref{eqn:oe-bound} can be nontrivial (stronger than $|\Delta S_M| \leq \log V$) only when $\rho,\sigma$ have support on an overlapping region of phase space. In this sense, the closeness of entropy ensured by the bound arises only from a probability to be in exactly the same state. Further one may observe that the bound will be time independent for isolated systems, whose dynamics are norm preserving. This can be considered in terms of the second law of thermodynamics for coarse-grained entropies. Suppose $\rho$ is concentrated at some particular ``second-law-violating'' point in phase space, one which evolves to a lower entropy state. Any $\sigma$ sufficiently close to $\rho$ in the norm distance must also be second-law-violating, which at first glance seems to violate the principle that such states are rare. Fortunately, by the above considerations, $\sigma$ may be concentrated in any phase space region disjoint from $\rho$, even arbitrarily nearby, without any restriction on its entropic dynamics. In this way the bound is consistent with the principle that nearby phase space points of a chaotic system may have drastically different dynamics.

\begin{acknowledgements}
The authors thank Ludovico Lami for comments regarding the measured relative entropy continuity and the minimax lemma, and Niklas Galke, Philipp Strasberg, and Giulio Gasbarri for helpful discussions.

The authors acknowledge support by MICIIN with funding from European Union NextGenerationEU (PRTR-C17.I1) and by Generalitat de Catalunya.
AW is furthermore supported by the European Commission QuantERA grant ExTRaQT (Spanish MICINN project PCI2022-132965), by the Spanish MINECO (project PID2019-107609GB-I00) with the support of FEDER funds, the Generalitat de Catalunya (project 2017-SGR-1127), by the Alexander von Humboldt Foundation, as well as the Institute for Advanced Study of the Technical University Munich. 
\end{acknowledgements}


\bibliography{biblio}

\clearpage

\end{document}